\documentclass{amsart}
\usepackage{amsfonts}
\usepackage{amsmath,amscd}
\usepackage{amsthm}
\usepackage{amssymb}
\usepackage{latexsym}

\newtheorem{cor}{Corollary}[section]
\newtheorem{lem}{Lemma}[section]

\newtheorem{prop}{Proposition}[section]

\newtheorem{defn}{Definition}[section]
\newtheorem{thm}{Theorem}

\newtheorem{rem}{Remark}
\newtheorem{notation}{Notation}

\numberwithin{equation}{section}
\begin{document}
\newcommand{\beqa}{\begin{eqnarray}}
\newcommand{\eeqa}{\end{eqnarray}}
\newcommand{\thmref}[1]{Theorem~\ref{#1}}
\newcommand{\secref}[1]{Sect.~\ref{#1}}
\newcommand{\lemref}[1]{Lemma~\ref{#1}}
\newcommand{\propref}[1]{Proposition~\ref{#1}}
\newcommand{\corref}[1]{Corollary~\ref{#1}}
\newcommand{\remref}[1]{Remark~\ref{#1}}
\newcommand{\er}[1]{(\ref{#1})}
\newcommand{\nc}{\newcommand}
\newcommand{\rnc}{\renewcommand}

\nc{\cal}{\mathcal}

\nc{\goth}{\mathfrak}
\rnc{\bold}{\mathbf}
\renewcommand{\frak}{\mathfrak}
\renewcommand{\Bbb}{\mathbb}

\newcommand{\id}{\text{id}}
\nc{\Cal}{\mathcal}
\nc{\Xp}[1]{X^+(#1)}
\nc{\Xm}[1]{X^-(#1)}
\nc{\on}{\operatorname}
\nc{\ch}{\mbox{ch}}
\nc{\Z}{{\bold Z}}
\nc{\J}{{\mathcal J}}
\nc{\C}{{\bold C}}
\nc{\Q}{{\bold Q}}
\renewcommand{\P}{{\mathcal P}}
\nc{\N}{{\Bbb N}}
\nc\beq{\begin{equation}}
\nc\enq{\end{equation}}
\nc\lan{\langle}
\nc\ran{\rangle}
\nc\bsl{\backslash}
\nc\mto{\mapsto}
\nc\lra{\leftrightarrow}
\nc\hra{\hookrightarrow}
\nc\sm{\smallmatrix}
\nc\esm{\endsmallmatrix}
\nc\sub{\subset}
\nc\ti{\tilde}
\nc\nl{\newline}
\nc\fra{\frac}
\nc\und{\underline}
\nc\ov{\overline}
\nc\ot{\otimes}
\nc\bbq{\bar{\bq}_l}
\nc\bcc{\thickfracwithdelims[]\thickness0}
\nc\ad{\text{\rm ad}}
\nc\Ad{\text{\rm Ad}}
\nc\Hom{\text{\rm Hom}}
\nc\End{\text{\rm End}}
\nc\Ind{\text{\rm Ind}}
\nc\Res{\text{\rm Res}}
\nc\Ker{\text{\rm Ker}}
\rnc\Im{\text{Im}}
\nc\sgn{\text{\rm sgn}}
\nc\tr{\text{\rm tr}}
\nc\Tr{\text{\rm Tr}}
\nc\supp{\text{\rm supp}}
\nc\card{\text{\rm card}}
\nc\bst{{}^\bigstar\!}
\nc\he{\heartsuit}
\nc\clu{\clubsuit}
\nc\spa{\spadesuit}
\nc\di{\diamond}
\nc\cW{\cal W}
\nc\cG{\cal G}
\nc\al{\alpha}
\nc\bet{\beta}
\nc\ga{\gamma}
\nc\de{\delta}
\nc\ep{\epsilon}
\nc\io{\iota}
\nc\om{\omega}
\nc\si{\sigma}
\rnc\th{\theta}
\nc\ka{\kappa}
\nc\la{\lambda}
\nc\ze{\zeta}

\nc\vp{\varpi}
\nc\vt{\vartheta}
\nc\vr{\varrho}

\nc\Ga{\Gamma}
\nc\De{\Delta}
\nc\Om{\Omega}
\nc\Si{\Sigma}
\nc\Th{\Theta}
\nc\La{\Lambda}

\nc\boa{\bold a}
\nc\bob{\bold b}
\nc\boc{\bold c}
\nc\bod{\bold d}
\nc\boe{\bold e}
\nc\bof{\bold f}
\nc\bog{\bold g}
\nc\boh{\bold h}
\nc\boi{\bold i}
\nc\boj{\bold j}
\nc\bok{\bold k}
\nc\bol{\bold l}
\nc\bom{\bold m}
\nc\bon{\bold n}
\nc\boo{\bold o}
\nc\bop{\bold p}
\nc\boq{\bold q}
\nc\bor{\bold r}
\nc\bos{\bold s}
\nc\bou{\bold u}
\nc\bov{\bold v}
\nc\bow{\bold w}
\nc\boz{\bold z}

\nc\ba{\bold A}
\nc\bb{\bold B}
\nc\bc{\bold C}
\nc\bd{\bold D}
\nc\be{\bold E}
\nc\bg{\bold G}
\nc\bh{\bold H}
\nc\bi{\bold I}
\nc\bj{\bold J}
\nc\bk{\bold K}
\nc\bl{\bold L}
\nc\bm{\bold M}
\nc\bn{\bold N}
\nc\bo{\bold O}
\nc\bp{\bold P}
\nc\bq{\bold Q}
\nc\br{\bold R}
\nc\bs{\bold S}
\nc\bt{\bold T}
\nc\bu{\bold U}
\nc\bv{\bold V}
\nc\bw{\bold W}
\nc\bz{\bold Z}
\nc\bx{\bold X}

\nc\ca{\mathcal A}
\nc\cb{\mathcal B}
\nc\cc{\mathcal C}
\nc\cd{\mathcal D}
\nc\ce{\mathcal E}
\nc\cf{\mathcal F}
\nc\cg{\mathcal G}
\rnc\ch{\mathcal H}
\nc\ci{\mathcal I}
\nc\cj{\mathcal J}
\nc\ck{\mathcal K}
\nc\cl{\mathcal L}
\nc\cm{\mathcal M}
\nc\cn{\mathcal N}
\nc\co{\mathcal O}
\nc\cp{\mathcal P}
\nc\cq{\mathcal Q}
\nc\car{\mathcal R}
\nc\cs{\mathcal S}
\nc\ct{\mathcal T}
\nc\cu{\mathcal U}
\nc\cv{\mathcal V}
\nc\cz{\mathcal Z}
\nc\cx{\mathcal X}
\nc\cy{\mathcal Y}

\nc\e[1]{E_{#1}}
\nc\ei[1]{E_{\delta - \alpha_{#1}}}
\nc\esi[1]{E_{s \delta - \alpha_{#1}}}
\nc\eri[1]{E_{r \delta - \alpha_{#1}}}
\nc\ed[2][]{E_{#1 \delta,#2}}
\nc\ekd[1]{E_{k \delta,#1}}
\nc\emd[1]{E_{m \delta,#1}}
\nc\erd[1]{E_{r \delta,#1}}

\nc\ef[1]{F_{#1}}
\nc\efi[1]{F_{\delta - \alpha_{#1}}}
\nc\efsi[1]{F_{s \delta - \alpha_{#1}}}
\nc\efri[1]{F_{r \delta - \alpha_{#1}}}
\nc\efd[2][]{F_{#1 \delta,#2}}
\nc\efkd[1]{F_{k \delta,#1}}
\nc\efmd[1]{F_{m \delta,#1}}
\nc\efrd[1]{F_{r \delta,#1}}

\nc\fa{\frak a}
\nc\fb{\frak b}
\nc\fc{\frak c}
\nc\fd{\frak d}
\nc\fe{\frak e}
\nc\ff{\frak f}
\nc\fg{\frak g}
\nc\fh{\frak h}
\nc\fj{\frak j}
\nc\fk{\frak k}
\nc\fl{\frak l}
\nc\fm{\frak m}
\nc\fn{\frak n}
\nc\fo{\frak o}
\nc\fp{\frak p}
\nc\fq{\frak q}
\nc\fr{\frak r}
\nc\fs{\frak s}
\nc\ft{\frak t}
\nc\fu{\frak u}
\nc\fv{\frak v}
\nc\fz{\frak z}
\nc\fx{\frak x}
\nc\fy{\frak y}

\nc\fA{\frak A}
\nc\fB{\frak B}
\nc\fC{\frak C}
\nc\fD{\frak D}
\nc\fE{\frak E}
\nc\fF{\frak F}
\nc\fG{\frak G}
\nc\fH{\frak H}
\nc\fJ{\frak J}
\nc\fK{\frak K}
\nc\fL{\frak L}
\nc\fM{\frak M}
\nc\fN{\frak N}
\nc\fO{\frak O}
\nc\fP{\frak P}
\nc\fQ{\frak Q}
\nc\fR{\frak R}
\nc\fS{\frak S}
\nc\fT{\frak T}
\nc\fU{\frak U}
\nc\fV{\frak V}
\nc\fZ{\frak Z}
\nc\fX{\frak X}
\nc\fY{\frak Y}
\nc\tfi{\ti{\Phi}}
\nc\bF{\bold F}
\rnc\bol{\bold 1}

\nc\ua{\bold U_\A}

%%%%%%%%%%%%%%%%%%%%%%%%%%%%%%%%%%%%%%%%%%%%%%%%%%%%%%
\nc\qinti[1]{[#1]_i}
\nc\q[1]{[#1]_q}
\nc\xpm[2]{E_{#2 \delta \pm \alpha_#1}}  %\xpm{j}{l}
\nc\xmp[2]{E_{#2 \delta \mp \alpha_#1}}
\nc\xp[2]{E_{#2 \delta + \alpha_{#1}}}
\nc\xm[2]{E_{#2 \delta - \alpha_{#1}}}
\nc\hik{\ed{k}{i}}
\nc\hjl{\ed{l}{j}}
\nc\qcoeff[3]{\left[ \begin{smallmatrix} {#1}& \\ {#2}& \end{smallmatrix}
\negthickspace \right]_{#3}}
\nc\qi{q}
\nc\qj{q}

\nc\ufdm{{_\ca\bu}_{\rm fd}^{\le 0}}

%%%%%%%%%%%%%%%%%%%%%%%%%%%%%%%%%%%%%%%%%%%%%%%%%%%%%%

%\nc\rtimes
\nc\isom{\cong} 

\nc{\pone}{{\Bbb C}{\Bbb P}^1}
\nc{\pa}{\partial}
\def\H{\mathcal H}
\def\L{\mathcal L}
\nc{\F}{{\mathcal F}}
\nc{\Sym}{{\goth S}}
\nc{\A}{{\mathcal A}}
\nc{\arr}{\rightarrow}
\nc{\larr}{\longrightarrow}

\nc{\ri}{\rangle}
\nc{\lef}{\langle}
\nc{\W}{{\mathcal W}}
\nc{\uqatwoatone}{{U_{q,1}}(\su)}
\nc{\uqtwo}{U_q(\goth{sl}_2)}
\nc{\dij}{\delta_{ij}}
\nc{\divei}{E_{\alpha_i}^{(n)}}
\nc{\divfi}{F_{\alpha_i}^{(n)}}
\nc{\Lzero}{\Lambda_0}
\nc{\Lone}{\Lambda_1}
\nc{\ve}{\varepsilon}
\nc{\phioneminusi}{\Phi^{(1-i,i)}}
\nc{\phioneminusistar}{\Phi^{* (1-i,i)}}
\nc{\phii}{\Phi^{(i,1-i)}}
\nc{\Li}{\Lambda_i}
\nc{\Loneminusi}{\Lambda_{1-i}}
\nc{\vtimesz}{v_\ve \otimes z^m}

\nc{\asltwo}{\widehat{\goth{sl}_2}}
\nc\ag{\widehat{\goth{g}}}  
\nc\teb{\tilde E_\boc}
\nc\tebp{\tilde E_{\boc'}}

\title[A new current algebra and the reflection equation]{A new current algebra and the reflection equation}
%\dedicatory{}
\author{P. Baseilhac}
\address{Laboratoire de Math\'ematiques et Physique Th\'eorique CNRS/UMR 6083,
           F\'ed\'eration Denis Poisson, Universit\'e de Tours, Parc de Grammont, 37200 Tours, FRANCE}
\email{baseilha@lmpt.univ-tours.fr}
\author{K. Shigechi}
\address{Institute for Theoretical Physics, Valckenierstraat 65, 1018 XE Amsterdam, THE NETHERLANDS}
\email{k.shigechi@uva.nl}
%\thanks{K.S is supported by the ANR Research project ``{\it Boundary integrable models: algebraic
%structures and correlation functions}'', contract number JC05-52749.}

%\pagestyle{plain} 

\begin{abstract}
We establish an explicit algebra isomorphism between the quantum reflection algebra for the $U_q(\widehat{sl_2})$ $R$-matrix and a new type of current algebra. These two algebras are shown to be two realizations of 
a special case of tridiagonal algebras ($q-$Onsager).
\end{abstract}

\maketitle

\vskip -0.2cm



\section{Introduction}
Discovered in the context of the quantum inverse scattering method for solving quantum integrable systems, quantum  groups appeared in the literature through different ways (see \cite{Chari} for references). On one hand, starting from the fundamental independent discovery of Drinfeld \cite{D1} and Jimbo \cite{J1} the quantum affine algebras $U_q({\widehat g})$ were initially formulated using a $q-$deformed version of the commutation relations between the elements of the Chevalley presentation of ${\widehat g}$. Later on \cite{D2}, Drinfeld proposed a new realization of $U_q({\widehat g})$ in terms of elements $\{x_{i,k}^{\pm},\varphi_{i,m},\psi_{i,n}|i=1,...,l;k\in{\mathbb Z},m\in-{\mathbb Z}_+,n\in{\mathbb Z}_+\}$ with $l=rank(g)$ generated through operator-valued functions $x^{\pm}_i(u),\varphi_i(u),\psi_i(u)$ of the formal variable $u$, the so-called currents. In some sense, the Drinfeld's realization is a quantum analogue of the loop realization of affine Lie algebras. Although Drinfeld stated the isomorphism between the two realizations, the proof only appeared later on \cite{Be,Jin}. In particular, in \cite{Be} (see also \cite{Dam}) Lusztig's theory of braid group action \cite{L} on the quantum enveloping algebras was used from which an explicit homomorphism from Drinfeld's new realization \cite{D2} to
the initial one \cite{D1,J1} was obtained.
On the other hand, an alternative realization of quantum affine algebras $U_q({\widehat g})$ by means of solutions of the quantum Yang-Baxter equation \cite{KRS,KS,F1} - called the $R-$matrix - and the ``RLL'' algebraic relations of the quantum inverse scattering method was proposed by Reshetikhin and Semenov-Tian-Shansky in \cite{RS}, extending the previous results of Faddeev-Reshetikhin-Takhtajan \cite{FRT1} for finite dimensional simple Lie algebra $g$.\vspace{1mm} 

In view of these realizations, in \cite{DF} Ding and Frenkel exhibited an explicit isomorphism between the ``RLL'' formulation and Drinfeld's second realization. Namely, $L-$operators were shown to admit a unique (Gauss) decomposition in terms of Drinfeld's currents $x^{\pm}_i(u),\varphi_i(u),\psi_i(u)$. So, all these different realizations may be summarized by the following picture which provides an unifying algebraic scheme for quantum affine algebras:
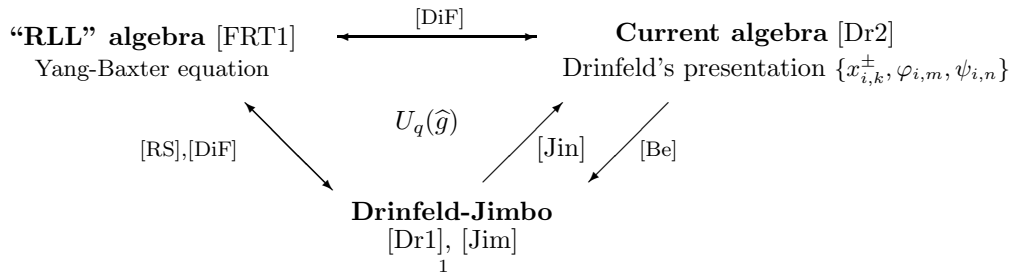
\begin{figure}[ht!]
\begin{center}
\begin{picture}(370,85)
   \put(20,60){\shortstack[1]{ \bf``RLL'' algebra \cite{FRT1}\\
        \small Yang-Baxter equation}}
   %\put(150,75){\vector(1,0){70}}
   \put(215,75){\vector(-2,0){70}}
   \put(150,75){\vector(2,0){70}}
   \put(140,80){\shortstack[l]{ \\
                                 \footnotesize \qquad \qquad \cite{DF}}}
   \put(230,60){\shortstack[l]{\qquad {\bf Current algebra} \cite{D2} \\
                                    Drinfeld's presentation \small $\{x_{i,k}^{\pm},\varphi_{i,m},\psi_{i,n}\}$}}
   \put(150,40){\shortstack[l]{\quad \ \ {\bf $U_q({\widehat{g}})$}}}                             
   \put(140,20){\vector(-1,1){30}}
   \put(113,47){\vector(1,-1){30}}
   \put(70,30){\shortstack[l]{\footnotesize \cite{RS},\cite{DF}}}
   \put(270,50){\vector(-1,-1){30}}
   \put(225,30){\shortstack[l]{\footnotesize \qquad \qquad \cite{Be}}}                             
   \put(200,20){\vector(1,1){30}}
   \put(220,30){\shortstack[l]{\cite{Jin}}}
   \put(140,-5){\shortstack[l]{\quad {\bf Drinfeld-Jimbo}\\
   															\qquad \ \cite{D1}, \cite{J1}}}
\end{picture}
\end{center}
%\vspace{3mm}
%\caption{\label{fig1}
%Realizations of $U_q({\widehat{sl_2}})$}
\end{figure} 
\vspace{1mm}

Beyond the interest of the algebraic structures involved, the explicit relation between the two different realizations (``RLL'' and Drinfeld's one) of $U_q({\widehat g})$ has found many interesting applications in the study of quantum integrable systems and representation theory.

For quantum integrable systems with boundaries, Cherednik \cite{Cher84} and later on \cite{Skly88} introduced another example of quadratic algebra associated with the so-called reflection equations. In this case, given an $R-$matrix associated with $U_q(\widehat{g})$ one is looking for a $K-$operator (sometimes called a Sklyanin's operator) satisfying ``RKRK'' algebraic relations. Motivated by the study of related integrable systems, several examples of $K-$operators acting on finite dimensional representations have been constructed. However, a formulation of $K-$operators in terms of current algebras i.e. a ``boundary'' - in reference to boundary integrable models - version of Ding-Frenkel \cite{DF} analysis has never been explicitly presented, nor a ``boundary''  analogue of Drinfeld's presentation even in the simplest case $U_q(\widehat{sl_2})$.\vspace{1mm}

In this paper, we argue that the $q-$Onsager algebra ${\mathbb T}$ (a type of tridiagonal algebra) which independently appeared in the context of orthogonal polynomial association schemes \cite{Ter03} and hidden symmetries of boundary integrable models \cite{B} admits analogously two alternative realizations. One realization is given in terms of a $K-$operator satisfying ``RKRK'' defining relations for the $U_q(\widehat{sl_2})$ $R$-matrix, and the other realization in terms of a new type of current algebra associated with the generating set $\{{\cW}_{-k},{\cW}_{k+1},{\cG}_{k+1},{\tilde{\cG}}_{k+1}|k\in{\mathbb Z}_+\}$ introduced in \cite{BK}. A new algebraic scheme follows, which extends to the family of reflection equation algebras the standard scheme relating the Faddeev-Reshetikhin-Takhtajan, Jimbo and Drinfeld (first and second) realizations of quantum affine algebras (see above picture). Although it is not considered here, the extension of our work to other classical Lie algebra - technically more complicated - is an interesting and open problem. \vspace{1mm}

The paper is organized as follows. In Section 2, a new current algebra - denoted $O_q(\widehat{sl_2})$ below - with generators $\cW_\pm(u),\cG_\pm(u)$ and formal variable $u$ is introduced. It is shown to be isomorphic to the ``RKRK'' algebra. A coaction map, the analogue of the coproduct for Hopf's algebras, is also explicitly derived. In Section 3, the new currents are found to be generating functions in the symmetric variable $U=(qu^2+q^{-1}u^{-2})/(q+q^{-1})$ which coefficients coincide with the elements of the infinite dimensional algebra - denoted ${\cal A}_q$ below -  introduced in \cite{BK}. In the last section, based on the commuting properties of the $K-$operator with the two generators of the $q-$Onsager algebra we establish the isomorphism between ${\mathbb T}$ and the ``RKRK'' algebra. A new algebraic scheme unifying these realizations is then proposed. 

\begin{notation}
In this paper, ${\mathbb R}$, ${\mathbb C}$, ${\mathbb Z}$ denote the field of real, complex numbers and integers, respectively. We denote ${\mathbb R}^*={\mathbb R}\backslash\{0\}$, ${\mathbb C}^*={\mathbb C}\backslash\{0\}$, ${\mathbb Z}^*={\mathbb Z}\backslash\{0\}$ and ${\mathbb Z}_+$ for nonnegative integers. We introduce the $q-$commutator $\big[X,Y\big]_q=qXY-q^{-1}YX$ where $q$ is the deformation parameter, assumed not to be a root of unity. 
\end{notation}

\section{A new current algebra}
Let ${\cal V}$ be a finite dimensional space. Let the operator-valued function $R:{\mathbb C}^*\mapsto \mathrm{End}({\cal V}\otimes {\cal V})$ be the intertwining operator (quantum $R-$matrix) between the tensor product of two fundamental representations ${\cal V}={\mathbb C}^2$ associated with the algebra $U_q(\widehat{sl_2})$. The element $R(u)$ depends on the deformation parameter $q$ and is defined by \cite{Baxter}
\begin{align}
R(u) =\left(
\begin{array}{cccc} 
 uq -  u^{-1}q^{-1}    & 0 & 0 & 0 \\
0  &  u -  u^{-1} & q-q^{-1} & 0 \\
0  &  q-q^{-1} & u -  u^{-1} &  0 \\
0 & 0 & 0 & uq -  u^{-1}q^{-1}
\end{array} \right) \ ,\label{R}
\end{align}
where $u$ is called the spectral parameter. Then $R(u)$ satisfies the quantum Yang-Baxter equation in the space ${\cal V}_1\otimes {\cal V}_2\otimes {\cal V}_3$. Using the standard notation $R_{ij}(u)\in \mathrm{End}({\cal V}_i\otimes {\cal V}_j)$, it reads 
\begin{align}
R_{12}(u/v)R_{13}(u)R_{23}(v)=R_{23}(v)R_{13}(u)R_{12}(u/v)\ \qquad \forall u,v.\label{YB}
\end{align}

Let us now consider an extension related with the reflection equation or boundary quantum Yang-Baxter equation - which was first introduced in the context of boundary quantum inverse scattering theory (see \cite{Cher84},\cite{Skly88} for details) -. For simplicity and without loosing generality we consider the simplest case, i.e. the $U_q(\widehat{sl_2})$ $R-$matrix defined above.

\begin{defn}[``RKRK'' Reflection equation algebra]{\label{defnRE}} Define $R(u)$ to be (\ref{R}). $B_q(\widehat{sl_2})$ is an associative algebra with unit $1$ and
generators $K_{11}(u)\equiv A(u)$, $K_{12}(u)\equiv B(u)$, $K_{21}(u)\equiv C(u)$, $K_{22}(u)\equiv D(u)$ considered as the elements of the $2\times 2$ square matrix $K(u)$ which obeys the defining relations $\forall u,v$
\begin{align} R_{12}(u/v)\ (K(u)\otimes I\!\!I)\ R_{12}(uv)\ (I\!\!I \otimes K(v))\
= \ (I\!\!I \otimes K(v))\ R_{12}(uv)\ (K(u)\otimes I\!\!I)\ R_{12}(u/v)\ .
\label{RE} \end{align}
\end{defn}

It is known that given a solution $K(u)$ of the reflection equation (\ref{RE}), one can construct by induction other solutions using suitable combinations of Lax operators $L(u)$. This is sometimes named as the ``dressing'' procedure. In particular, for the simplest case one has:
\begin{prop}[see \cite{Skly88}]\label{propL} Given $R(u)$ defined by (\ref{R}),
let $L(u)$ be a solution of the quantum Yang-Baxter algebra with defining relations $\forall u,v$
\beqa R(u/v)(L(u)\otimes I\!\!I)(I\!\!I \otimes L(v))\
= \ (I\!\!I \otimes L(v)) (L(u)\otimes I\!\!I) R(u/v)\ .
\label{YB} \eeqa
Let $K(u)$ be a solution of (\ref{RE}). Then, the matrix $L(u)K(u)L^{-1}(u^{-1})$ is a solution of the reflection equation (\ref{RE}).
\end{prop}
For instance, using the generating set $\{S_\pm,s_3\}$ of the quantum algebra $U_{q}(sl_2)$ with defining relations
$[s_3,S_\pm]=\pm S_\pm$ and 
$[S_+,S_-]=(q^{2s_3}-q^{-2s_3})/(q-q^{-1})$\ , it is known that the Lax operator
\beqa
{L}(u) =\left(
\begin{array}{cc}
 uq^{{1\over 2}}q^{s_3}- u^{-1}q^{-{1\over 2}}q^{-s_3}    &(q-q^{-1})S_-\\
(q-q^{-1})S_+    &  uq^{{1\over 2}}q^{-s_3}- u^{-1}q^{-{1\over 2}}q^{s_3}  \\
\end{array} \right) \ \label{Lax}
\eeqa
satisfies the quantum Yang-Baxter algebra (\ref{YB}). 
In quantum integrable lattice models with boundaries, the ``dressing'' procedure is often used. Starting from an elementary solution with $c-$number entries (associated with one boundary of the system) and dressing the $K-$operator with a product of $N$ $L-$operators acting on different quantum spaces, one reconstructs a whole spin chain with $N$ sites including inhomogeneities, if necessary \cite{Skly88}.\vspace{2mm} 

In order to exhibit the new current algebra starting from the ``RKRK'' reflection equation algebra, based on previous works on boundary quantum integrable systems on the lattice \cite{B,BK} it seems rather natural to write the elements $A(u)$, $B(u)$, $C(u)$, $D(u)$ in terms of new currents as follows. It may be important to stress that Proposition \ref{propL} plays an essential role (see \cite{B,BK}) in suggesting such a decomposition.
 
\begin{lem}{\label{lem}} Suppose $q\neq 1$, $u\neq q^{-1}$ and $k_\pm\in {\mathbb C}^*$. 
Any solution of the reflection equation algebra $B_q(\widehat{sl_2})$ admits the following decomposition in terms of new elements $\cW_\pm(u)$, $\cG_\pm(u)$:
\begin{align}
A(u)= uq \cW_+(u) - u^{-1}q^{-1} \cW_-(u)\ ,\label{m1}\\
D(u)= uq \cW_-(u) - u^{-1}q^{-1} \cW_+(u)\ ,\label{m2}\\
B(u)= \frac{1}{k_-(q+q^{-1})}\cG_+(u) + \frac{k_+(q+q^{-1})}{(q-q^{-1})}\ ,\label{m3}\\
C(u)= \frac{1}{k_+(q+q^{-1})}\cG_-(u) + \frac{k_-(q+q^{-1})}{(q-q^{-1})}\ .\label{m4}
\end{align}
Given the elements $A(u),B(u),C(u)$ of this form, this decomposition is unique.
\end{lem}
\begin{proof}
The reflection equation being satisfied for arbitrary $u,v\in{\mathbb C}^*$ and generic $q$, in view of (\ref{R}) the elements $A(u)$, $B(u)$, $C(u)$, $D(u)$ are {\it a priori} formal power series in $u$. With no restrictions, let us choose $A(u)$, $B(u)$, $C(u)$ to be (\ref{m1}), (\ref{m3}), (\ref{m4}), respectively. We have to show that $D(u)$ is uniquely defined by (\ref{m2}). To prove it, assume the set $\{A,B,C,D\}$ given by (\ref{m1})-(\ref{m4}) satisfies the reflection equation algebra with (\ref{R}). In terms of these elements, explicitly (\ref{RE}) reads
\begin{eqnarray*}
 (i)&&a_-c_+
 \left(B C' -B' C \right)+a_-a_+
 [A ,A']=0\ ,\\
 (i')&&a_-c_+
 \left(C B' -C' B \right)+a_-a_+
 [D ,D']=0\ ,\\
(ii) &&b_-b_+
 [A, D']+
 c_-c_+
 [D ,D']+\ c_-a_+
 \big(C B' -C' B \big)=0\ ,\\
 (ii') &&b_-b_+
 [D, A']+
 c_-c_+
 [A ,A']+\ c_-a_+
 \big(B C' -B' C \big)=0\ ,\\
(iii)&&c_-b_+
 \big(D A'-D'A \big)+ b_-c_+
 \big(A A'-D'D\big)+\ b_-a_+
 [B,C']=0\ ,\\
(iii')&&c_-b_+
 \big(AD'-A'D \big)+ b_-c_+
 \big(DD'-A'A\big)+\ b_-a_+
 [C,B']=0\ ,\\
(iv)&&b_-b_+AC'  +
 c_-c_+DC'+\ c_-a_+CA'
 -\ a_-a_+C' A -\ a_-c_+D'C=0\ ,\\
 (v)&&b_-b_+B'A+\ c_-c_+ B'D+\ c_-a_+A'B-\ a_-a_+A B'-\ a_-c_+BD'= 0\ ,\\
(vi)&&b_-b_+C'D+\ c_-c_+C'A+\ c_-a_+D'C-\ a_-a_+DC'-\ a_-c_+CA'=0\ ,\\
(vii)&&b_-b_+DB'+\ c_-c_+AB'+\ c_-a_+BD'-\ a_-a_+B'D-\ a_-c_+A'B= 0\ ,\\
(viii)&&b_-a_+BD' +
 \ c_-b_+DB'+\ b_-c_+AB'
 -\ a_-b_+D'B=0\ ,\\
(ix)&&b_-a_+A'B+\ c_-b_+B'A+\ b_-c_+B'D-\ a_-b_+BA'=0\ ,\\
(x)&&b_-a_+D'C+\ c_-b_+C'D+\ b_-c_+C'A-\ a_-b_+CD'=0\ ,\\
(xi)&&b_-a_+CA'+\ c_-b_+AC'+\ b_-c_+DC'-\ a_-b_+A'C=0\ ,\\
(xii)&& a_-b_+[B,B']=0\ ,\\
(xiii)&& a_-b_+[C,C']=0\ ,
\end{eqnarray*}
where $a(u)=uq-u^{-1}q^{-1}$, $b(u)=u-u^{-1}$, $c_\pm=q-q^{-1}$ and
we used the shorthand notations $a_-=a(u/v)$,
$a_+=a(uv)$ and similarly for $b$. Also 
$A=A(u)$, $A'=A(v)$ and similarly for $B,C$ and
$D$. Now, consider another set, say $\{A,B,C,{\overline D}\}$, ${\overline D(u)}=D(u)+f(u)$ where $f(u)$ is an unknown  function of $u$ and the elements of the reflection equation algebra. If $\{A,B,C,{\overline D}\}$ is also a solution of the reflection equation algebra, then $f(u)\equiv f(A,B,C,D;u)$ - the equations $(i)-(xiii)$ being the complete set of defining relations. Replacing  ${\overline D(u)}$ in $(iv)-(xi)$, 
we obtain $B(u)f(v) = f(u)B(v) = C(u)f(v) = f(u)C(v) = 0$ $\forall u,v$. On the other hand, from $(i)-(iii')$ one gets
$\big[A(u),f(v)\big]=0$. Acting with the l.h.s of $(ix)$ on $f(w)$ and using previous equations it follows $\big[D(u),f(w)\big]=0$ $\forall u,w$. All these equations imply that $f(u)\equiv 0$ $\forall u$.
\end{proof}

The next step is to prove the equivalence between the (sixteen in total) independent equations coming from the reflection equation algebra (\ref{RE}) with (\ref{R}) and a closed system of commutation relations among the currents. The relations below are among the main results of the paper. 
%\vspace{-1mm}

\begin{defn}[Current algebra]{\label{defnCA}} $O_q(\widehat{sl_2})$ is an associative algebra with unit $1$, current generators $\cW_\pm(u)$, $\cG_\pm(u)$ and parameter $\rho\in{\mathbb C}^*$. Define the formal variables $U=(qu^2+q^{-1}u^{-2})/(q+q^{-1})$ and $V=(qv^2+q^{-1}v^{-2})/(q+q^{-1})$ \ $\forall u,v$. The defining relations are:
\begin{align}
&&\big[{\cW}_\pm(u),{\cW}_\pm(v)\big]=0\ ,\qquad\qquad\qquad\qquad\qquad\qquad\qquad\label{ec1}\\
&&\big[{\cW}_+(u),{\cW}_-(v)\big]+\big[{\cW}_-(u),{\cW}_+(v)\big]=0\ ,\qquad\qquad\qquad\qquad\qquad\qquad\qquad\label{ec3}\\
&&(U-V)\big[{\cW}_\pm(u),{\cW}_\mp(v)\big]= \frac{(q-q^{-1})}{\rho(q+q^{-1})}\left({\cG}_\pm(u){\cG}_\mp(v)-{\cG}_\pm(v){\cG}_\mp(u)\right)\qquad\qquad\qquad\label{ec4}\\
&& \qquad \qquad\qquad\qquad\qquad\qquad\qquad\qquad+ \frac{1}{(q+q^{-1})} \big({\cG}_\pm(u)-{\cG}_\mp(u)+{\cG}_\mp(v)-{\cG}_\pm(v)\big)\ ,\nonumber
\end{align}
\beqa
%\begin{align}
&&{\cW}_\pm(u){\cW}_\pm(v)-{\cW}_\mp(u){\cW}_\mp(v)+\frac{1}{\rho(q^2-q^{-2})}\big[{\cG}_\pm(u),{\cG}_\mp(v)\big]\qquad\qquad\qquad\qquad \qquad\label{ec5}\\
&&\qquad\qquad\qquad\qquad\qquad\qquad+ \ \frac{1-UV}{U-V}\big({\cW}_\pm(u){\cW}_\mp(v)-{\cW}_\pm(v){\cW}_\mp(u)\big)=0\ ,\nonumber\\
&&U\big[{\cG}_\mp(v),{\cW}_\pm(u)\big]_q -V\big[{\cG}_\mp(u),{\cW}_\pm(v)\big]_q - (q-q^{-1})\big({\cW}_\mp(u){\cG}_\mp(v)-{\cW}_\mp(v){\cG}_\mp(u)\big)\label{ec6}\\
&&\qquad\qquad\qquad\qquad\qquad\qquad\qquad \quad + \ \rho \big(U{\cW}_\pm(u)-V{\cW}_\pm(v)-{\cW}_\mp(u)+{\cW}_\mp(v)\big)=0\ ,\nonumber\\
&&U\big[{\cW}_\mp(u),{\cG}_\mp(v)\big]_q -V\big[{\cW}_\mp(v),{\cG}_\mp(u)\big]_q - (q-q^{-1})\big({\cW}_\pm(u){\cG}_\mp(v)-{\cW}_\pm(v){\cG}_\mp(u)\big)\label{ec7}\\
&& \qquad\qquad\qquad\qquad\qquad\qquad\qquad\quad+  \ \rho \big(U{\cW}_\mp(u)-V{\cW}_\mp(v)-{\cW}_\pm(u)+{\cW}_\pm(v)\big)=0\ ,\nonumber\\
&&\big[{\cG}_\epsilon(u),{\cW}_\pm(v)\big]+\big[{\cW}_\pm(u),{\cG}_\epsilon(v)\big]=0 \ ,\quad \forall \epsilon=\pm\label{ec8}\ ,\qquad\qquad\qquad\qquad\qquad\qquad\qquad\\
&&\big[{\cG}_\pm(u),{\cG}_\pm(v)\big]=0\ ,\label{ec9}\qquad\qquad\qquad\qquad\qquad\qquad\qquad\\ 
&&\big[{\cG}_+(u),{\cG}_-(v)\big]+\big[{\cG}_-(u),{\cG}_+(v)\big]=0\ .\qquad\qquad\qquad\qquad\qquad\qquad\qquad\ \label{ec16}
%\end{align}
\eeqa
\end{defn}

\begin{rem} There exists an automorphism $\Omega$ defined by:
\beqa
\Omega(\cW_\pm(u))=\cW_\mp(u) \ , \qquad  \Omega(\cG_\pm(u))=\cG_\mp(u)\ . 
\eeqa
\end{rem}

Contrary to all known examples of Drinfeld currents associated with quantum affine Lie algebras or superalgebras, it is important to notice that the variables $u,v$ only arise through the symmetric $(qx^2\leftrightarrow q^{-1}x^{-2},\  \forall x\in u,v$) combinations $U,V$, respectively. 
In view of the connections with algebraic structures that appear in boundary quantum integrable models \cite{B,B2}, such a fact is not surprising although not obvious from (\ref{RE}). We now turn to the derivation of all equations above.

\begin{thm}\label{th1}
The map $\Phi: B_q(\widehat{sl_2}) \mapsto O_q(\widehat{sl_2})$ defined by (\ref{m1}-\ref{m4}) is an algebra isomorphism.
\end{thm}

\begin{proof} 
First, according to Lemma \ref{lem} we have to show that the map $\Phi$ defined by (\ref{m1}-\ref{m4}) is an algebra homomorphism from $B_q(\widehat{sl_2})$ to $O_q(\widehat{sl_2})$. Set $\rho\equiv k_+k_-(q+q^{-1})^2$ 
and define
\beqa
&&\qquad \qquad X_1\equiv \big[{\cW}_+(u),{\cW}_+(v)\big]\ ,\qquad X_2\equiv \big[{\cW}_-(u),{\cW}_-(v)\big]\ ,\nonumber \\
&&\qquad \qquad X_3\equiv \big[{\cW}_+(u),{\cW}_-(v)\big] + \big[{\cW}_-(u),{\cW}_+(v)\big]\ ,\nonumber\\
&&\qquad \qquad X_4\equiv \big[{\cG}_+(u),{\cG}_-(v)\big] + \big[{\cG}_-(u),{\cG}_+(v)\big]\ ,\nonumber\\
&&\qquad \qquad X_5\equiv(q+q^{-1})(U-V)\big[{\cW}_+(u),{\cW}_-(v)\big]- \frac{(q-q^{-1})}{\rho}\left({\cG}_+(u){\cG}_-(v)-{\cG}_+(v){\cG}_-(u)\right)\qquad\qquad\qquad\nonumber\\
&&\qquad \qquad\qquad\qquad\qquad\qquad\qquad\qquad -\big({\cG}_+(u)-{\cG}_-(u)+{\cG}_-(v)-{\cG}_+(v)\big)\ , \nonumber
\eeqa
where the variables $U\equiv (qu^2+q^{-1}u^{-2})/(q+q^{-1})$ and similarly for $V$ are introduced.
In terms of the combinations $X_i$, it is straightforward to show that the equations $(i),(i'),(ii),(ii')$ above can be simply written, respectively, as
\beqa
(i) && \Leftrightarrow \quad a(uv)uvq^2X_1 + a(uv)u^{-1}v^{-1}q^{-2}X_2 - a(uv)u^{-1}vX_3 - X_5=0 \ ,\nonumber\\
(i') &&\Leftrightarrow \quad a(uv)uvq^2X_2 + a(uv)u^{-1}v^{-1}q^{-2}X_1 - a(uv)uv^{-1}X_3 + \frac{q-q^{-1}}{\rho}X_4 + X_5=0\ ,\nonumber
\eeqa
\beqa
(ii) && \Leftrightarrow \quad 
\big(b(u/v)b(uv)uv^{-1} -(q-q^{-1})^2u^{-1}v^{-1}q^{-2}\big)X_1\nonumber\\
&&\qquad +\big(b(u/v)b(uv)u^{-1}v -(q-q^{-1})^2uvq^2\big)X_2 \nonumber\\
&&\qquad- \big(b(u/v)b(uv)u^{-1}v^{-1}q^{-2} -(q-q^{-1})^2uv^{-1}\big)X_3\nonumber\\
&&\qquad - a(uv)\frac{(q-q^{-1})}{\rho}X_4 - a(uv)X_5=0 \ ,\nonumber
\eeqa
\beqa
(ii') && \Leftrightarrow \quad 
\big(b(u/v)b(uv)u^{-1}v -(q-q^{-1})^2uvq^2\big)X_1\nonumber\\
&&\qquad +\big(b(u/v)b(uv)uv^{-1} -(q-q^{-1})^2u^{-1}v^{-1}q^{-2}\big)X_2 \nonumber\\
&&\qquad- \big(b(u/v)b(uv)uvq^2 -(q-q^{-1})^2u^{-1}v\big)X_3\nonumber\\
&&\qquad - a(uv)X_5=0 \ .\nonumber
\eeqa
Simplifying these expressions, in particular it follows
\beqa
a(uv)(i)-(ii') && \Leftrightarrow \quad v^2q^2X_1+v^{-2}q^{-2}X_2-X_3=0\ ,\nonumber\\
a(uv)(i')-(ii) && \Leftrightarrow \quad v^2q^2X_2+v^{-2}q^{-2}X_1-X_3=0\ .\nonumber 
\eeqa
Considering both equations for $v$ arbitrary, it implies $X_1=X_2$. Then it is important to notice that the combinations $X_i|_{u\leftrightarrow v}=-X_i$ for $i=1,2,3$. As now $X_3=(v^2q^2+v^{-2}q^{-2})X_1$ and $u$ is arbitrary, it immediately follows $X_3\equiv X_1\equiv X_2\equiv 0$. Plugged into $(ii)$, $(ii')$ we obtain $X_4\equiv X_5\equiv 0$. In terms of the currents, these equalities lead to the commutation relations (\ref{ec1}), (\ref{ec3}), (\ref{ec4}), (\ref{ec16}). 

As a consequence of these relations, after some straightforward calculations one finds that the equations $(iii),(iii')$ drastically simplify into the relations (\ref{ec5}).

Let us now consider the equations $(iv),(vi),(x),(xi)$ above. Proceeding similarly, let us introduce 
\beqa
&&Y_1\equiv (q+q^{-1})\big(U\big[C(v),{\cW}_+(u)\big]_q - V\big[C(u),{\cW}_+(v)\big]_q + (q-q^{-1})\big({\cW}_-(v)C(u)-{\cW}_-(u)C(v)   \big)\big)\ ,\nonumber\\
&&Y_2\equiv (q+q^{-1})\big(U\big[{\cW}_-(u),C(v)\big]_q - V\big[{\cW}_-(v),C(u)\big]_q + (q-q^{-1})\big({\cW}_+(v)C(u)-{\cW}_+(u)C(v)   \big)\big)\ ,\nonumber\\
&&Y_3\equiv \big[C(u),{\cW}_+(v)\big]+\big[{\cW}_+(u),C(v)\big]\ ,\nonumber\\
&&Y_4\equiv \big[C(u),{\cW}_-(v)\big]+\big[{\cW}_-(u),C(v)\big]\ .\nonumber
\eeqa
In terms of these combinations, the equations $(iv),(vi),(x),(xi)$ become, respectively,
\beqa
(iv) && \Leftrightarrow \quad u\big(qY_1+q(v^2+v^{-2})Y_3+(q-q^{-1})Y_4)\big)\nonumber\\
 &&\qquad \quad  + \ u^{-1}\big(q^{-1}Y_2-q^{-1}(v^2+v^{-2})Y_4+(q-q^{-1})Y_3)\big) = 0\ ,\nonumber\\
(vi) && \Leftrightarrow \quad u\big(qY_2-q(v^2+v^{-2})Y_4+q^2(q-q^{-1})Y_3)\big)\nonumber\\
 &&\qquad \quad  + \ u^{-1}\big(q^{-1}Y_1+q^{-1}(v^2+v^{-2})Y_3+q^{-2}(q-q^{-1})Y_4)\big) = 0\ ,\nonumber\\
(x) && \Leftrightarrow \quad v\big(Y_2-q(q+q^{-1})UY_4+(q^2-q^{-2})Y_3)\big)\nonumber\\
 &&\qquad \quad  + \ v^{-1}\big(Y_1+q^{-1}(q+q^{-1})UY_3+(q^2-q^{-2})Y_4)\big) = 0\ ,\nonumber\\
(xi) && \Leftrightarrow \quad v\big(Y_1+q(q+q^{-1})UY_3)\big)\nonumber\\
 &&\qquad \quad  + \ v^{-1}\big(Y_2-q^{-1}(q+q^{-1})UY_4)\big) = 0\ .\nonumber
\eeqa
The variables $u,v$ and deformation parameter $q$ being arbitrary, compatibility of these equations implies $Y_1\equiv Y_2\equiv Y_3\equiv Y_4\equiv 0$. Replacing the explicit expression of $C(u)$ into $Y_i$, one ends up with the commutation relations (\ref{ec6}), (\ref{ec7}), (\ref{ec8}) for the current ${\cG}_-(u)$. Similar analysis for the remaining equations $(v),(vii),(viii),(ix)$ imply (\ref{ec6}), (\ref{ec7}), (\ref{ec8}) for ${\cG}_+(u)$. Finally, from $(xii),(xiii)$ we immediately obtain (\ref{ec9}). 

Surjectivity of the map being shown, the injectivity of the homomorphism follows from the fact that $\Phi$ is invertible for $u$ generic. This completes the proof.
\end{proof}

Quantum affine algebras are known to be Hopf algebras, thanks to the existence of a coproduct, counit and antipode actions. Although the explicit Hopf algebra isomorphism between Drinfeld's new realization (currents) and Drinfeld-Jimbo construction is still an open problem, several results are already known (see for instance \cite{DF}). 
For the new current algebra (\ref{ec1})-(\ref{ec16}), it is also important to exhibit analogous properties. Actually, solely using the results of \cite{Skly88} a coaction map \cite{Chari} can be easily identified. 
\begin{prop}\label{propdelta} For any $k_\pm,w\in {\mathbb C}^*$, there exists an algebra homomorphism $\delta_w: O_q(\widehat{sl_2})\mapsto U_q(sl_2)\times O_q(\widehat{sl_2})$ such that
\beqa
\delta_w(\cW_\pm(u))\!\!&=&\!\!\!\left((q-q^{-1})^2S_\pm S_\mp -q(q^{\pm 2s_3}-q^{\mp 2s_3}) \right)\otimes
{\cW}_\mp(u)
-(w^2+w^{-2})I\!\!I\otimes {\cW}_{\pm}(u) 
\nonumber\\
%\noalign{\vskip 0.1cm}
&&\!\!\!\!+ \frac{(q-q^{-1})}{k_+k_-(q+q^{-1})}
\left(k_+w^{\pm 1}q^{\pm 1/2}S_+q^{\pm s_3}\otimes
{\cG}_{+}(u)+k_-w^{\mp 1}q^{\mp 1/2}S_-q^{\pm s_3}\otimes {\cG}_{-}(u)\right)\nonumber\\
&&\!\!\!\!+ (q+q^{-1})\left((k_+w^{\pm 1}q^{\pm 1/2}S_+q^{\pm s_3}+k_-w^{\mp 1}q^{\mp 1/2}S_-q^{\pm s_3})\otimes I\!\!I + q^{\pm 2s_3}\otimes U{\cW}_{\pm}(u)\right) ,\nonumber\\
%%%%%%%%%%%%%%%%%%%%%%%%%%% 
\delta_w(\cG_\pm(u))&=&\!\!\!\frac{k_\mp}{k_\pm}(q-q^{-1})^2S_\mp^2\otimes
{\cG}_\mp(u)
-(w^2q^{\pm 2s_3}+w^{-2}q^{\mp 2s_3})\otimes {\cG}_{\pm}(u) + I\!\!I\otimes (q+q^{-1})U{\cG}_{\pm}(u) 
\nonumber\\
%\noalign{\vskip 0.1cm}
&&\ \ + \ (q+q^{-1})^2(q-q^{-1})\left(k_\mp w^{\pm 1}q^{\mp 1/2}S_\mp q^{ s_3}\otimes
(U{\cW}_{+}(u)-{\cW}_{-}(u))\right.\nonumber \\
&& \qquad \qquad \qquad \qquad \qquad \quad + \left.k_\mp w^{\mp 1}q^{\pm 1/2}S_\mp q^{- s_3}\otimes (U{\cW}_{-}(u)-{\cW}_{+}(u))\right)
\nonumber\\
&&\ \ + \ \frac{k_+k_-(q+q^{-1})^2}{(q-q^{-1})}\left((q+q^{-1})U +\frac{k_\mp}{k_\pm}(q-q^{-1})^2S_\mp^2 - (w^2q^{\pm 2s_3}+w^{-2}q^{\mp 2s_3}+1)  \right)\otimes I\!\!I\nonumber\ .
\eeqa
\end{prop}
\begin{proof}
According to \cite{Skly88} (see Proposition \ref{propL}) and the Lax operator (\ref{Lax}), $L(uw)K(u)L(uw^{-1})$ is a solution $\forall w$ of (\ref{RE}). Expanding this expression using (\ref{m1})-(\ref{m4}), the new entries of $L(uw)K(u)L(uw^{-1})$ are found to take the form (\ref{m1})-(\ref{m4}) replacing $\cW_\pm(u)\rightarrow \delta_w(\cW_\pm(u))$,  $\cG_\pm(u)\rightarrow \delta_w(\cG_\pm(u))$. For more details, we refer the reader to \cite{BK} where  similar calculations have been performed.
\end{proof}

\section{Another presentation}
In \cite{BK}, an infinite dimensional algebra denoted below ${\cal A}_q$ was proposed in order to solve boundary integrable systems with hidden symmetries related with a coideal subalgebra of $U_q(\widehat{sl_2})$. However, its defining relations were essentially conjectured based on the commutation relations and properties of certain operators acting on irreducible finite dimensional tensor product of evaluation representations. The aim of this Section is to construct an analogue of Drinfeld's presentation for the current algebra (\ref{ec1})-(\ref{ec16}). As a consequence, it provides a rigorous derivation of the relations conjectured in \cite{BK}.
\begin{defn}[\cite{BK}] ${\cal A}_q$ is an associative algebra with parameter $\rho\in{\mathbb C}^*$, unit $1$ and generators $\{{\cW}_{-k},{\cW}_{k+1},{\cG}_{k+1},{\tilde{\cG}}_{k+1}|k\in {\mathbb Z}_+\}$ satisfying the following relations:
\begin{align}
\big[{\cW}_0,{\cW}_{k+1}\big]=\big[{\cW}_{-k},{\cW}_{1}\big]=\frac{1}{(q+q^{-1})}\big({\tilde{\cG}_{k+1} } - {{\cG}_{k+1}}\big)\ ,\label{qo1}\\
\big[{\cW}_0,{\cG}_{k+1}\big]_q=\big[{\tilde{\cG}}_{k+1},{\cW}_{0}\big]_q=\rho{\cW}_{-k-1}-\rho{\cW}_{k+1}\ ,\label{qo2}\\
\big[{\cG}_{k+1},{\cW}_{1}\big]_q=\big[{\cW}_{1},{\tilde{\cG}}_{k+1}\big]_q=\rho{\cW}_{k+2}-\rho{\cW}_{-k}\ ,\label{qo3}\\
\big[{\cW}_{-k},{\cW}_{-l}\big]=0\ ,\quad 
\big[{\cW}_{k+1},{\cW}_{l+1}\big]=0\ ,\label{qo4}\quad \\
\big[{\cW}_{-k},{\cW}_{l+1}\big]
+\big[{{\cW}}_{k+1},{\cW}_{-l}\big]=0\ ,\label{qo5}\\
\big[{\cW}_{-k},{\cG}_{l+1}\big]
+\big[{{\cG}}_{k+1},{\cW}_{-l}\big]=0\ ,\label{qo6}\\
\big[{\cW}_{-k},{\tilde{\cG}}_{l+1}\big]
+\big[{\tilde{\cG}}_{k+1},{\cW}_{-l}\big]=0\ ,\label{qo7}\\
\big[{\cW}_{k+1},{\cG}_{l+1}\big]
+\big[{{\cG}}_{k+1},{\cW}_{l+1}\big]=0\ ,\label{qo8}\\
\big[{\cW}_{k+1},{\tilde{\cG}}_{l+1}\big]
+\big[{\tilde{\cG}}_{k+1},{\cW}_{l+1}\big]=0\ ,\label{qo9}\\
\big[{\cG}_{k+1},{\cG}_{l+1}\big]=0\ ,\quad   \big[{\tilde{\cG}}_{k+1},\tilde{{\cG}}_{l+1}\big]=0\ ,\label{qo10}\\
\big[{\tilde{\cG}}_{k+1},{\cG}_{l+1}\big]
+\big[{{\cG}}_{k+1},\tilde{{\cG}}_{l+1}\big]=0\ .\label{qo11}
\end{align}
\end{defn}

A natural ordering of ${\cal A}_q$ arises from the study of the commutation relations above.
Indeed, starting from monomials of lowest  $k=0,1,...$  and using (\ref{qo1}) possible definitions of ${\cG}_1,{\tilde{\cG}}_{1}$ are such that $\mathrm{d}[{\cG}_1]=\mathrm{d}[{\tilde{\cG}}_{1}]\leq 2$, where $\mathrm{d}$ denotes the degree of the monomials in the elements ${\cW}_0,{\cW}_1$. By induction, from (\ref{qo2}), (\ref{qo3}) with (\ref{qo1}) one immediately gets:
\begin{cor} The elements of ${\cal A}_q$ are monomials in ${\cW}_0,{\cW}_1$ of degree:
\beqa
\qquad \qquad \mathrm{d}[{\cW}_{-k}]=\mathrm{d}[{{\cW}}_{k+1}]\leq 2k+1 \qquad \mbox{and} \qquad \mathrm{d}[{\cG}_{k+1}]=\mathrm{d}[{\tilde{\cG}}_{k+1}]\leq 2k+2 \ , \qquad k\in{\mathbb Z}_+. \label{order}
\eeqa
\end{cor}
Note that writing explicitly all higher elements of ${\cal A}_q$ in terms of ${\cW}_0,{\cW}_1$ is essentially related with the construction of a Poincare-Birkoff-Witt basis for the algebra considered in the next Section, a problem that will be considered elsewhere. 

\begin{rem} According to the ordering (\ref{order}), the elements ${\cG}_{1},{\tilde{\cG}}_{1}\in{\cal A}_q$ are uniquely determined:
\beqa
{\cG}_{1} = \big[{\cW}_{1},{\cW}_{0}\big]_q + \alpha \qquad \mbox{and} \qquad  {\tilde{\cG}}_{1}=\big[{\cW}_{0},{\cW}_{1}\big]_q +  \alpha \qquad \forall  \alpha\ \in {\mathbb C}\ .\label{g1} 
\eeqa
\end{rem}

For the derivation of the second theorem, several other equalities will be required which can all be deduced from the relations above and (\ref{g1}). Indeed, let us show the following. 
\begin{prop}{\label{prop3}} If (\ref{qo1})-(\ref{qo11}) are satisfied, then the following relations hold:
\beqa
&&\qquad \qquad \quad \big[{\cW}_{-k-1},{\cW}_{l+1}\big] - \big[{\cW}_{-k},{\cW}_{l+2}\big] = \frac{q-q^{-1}}{\rho(q+q^{-1})}\big({\cG}_{k+1}\tilde{{\cG}}_{l+1}-{\cG}_{l+1}\tilde{{\cG}}_{k+1}\big)\ ,\label{h1}\\
&& \qquad \qquad \quad-{\cW}_{-k}{\cW}_{0} + {\cW}_{k+1}{\cW}_{1} - {\cW}_{-k-1}{\cW}_{1} + {\cW}_{0}{\cW}_{k+2} - \frac{1}{\rho(q^2-q^{-2})}\big[{\cG}_{k+1},{\tilde{\cG}}_{1}\big]=0\ ,\label{h2}
\eeqa
\beqa
&&\qquad \qquad \quad{\cW}_{-k-1}{\cW}_{-l} - {\cW}_{k+2}{\cW}_{l+1} - {\cW}_{-k}{\cW}_{-l-1} + {\cW}_{k+1}{\cW}_{l+2}\label{h3}\\
&&\qquad \qquad \quad+{\cW}_{-k}{\cW}_{l+1} - {\cW}_{-l}{\cW}_{k+1} - {\cW}_{-k-1}{\cW}_{l+2} + {\cW}_{-l-1}{\cW}_{k+2}\nonumber\\
&&\qquad \qquad \qquad \qquad \quad + \frac{1}{\rho(q^2-q^{-2})}
\big(\big[{\cG}_{k+2},{\tilde{\cG}}_{l+1}\big]-\big[{\cG}_{k+1},{\tilde{\cG}}_{l+2}\big]\big)=0\ ,\nonumber
\eeqa
\beqa
&&\qquad \qquad \quad \big[{\cG}_{l+1},{\cW}_{k+2}\big]_q - \big[{\cG}_{k+1},{\cW}_{l+2}\big]_q - (q-q^{-1})\big({\cW}_{-k}{\cG}_{l+1}-{\cW}_{-l}{\cG}_{k+1}\big)=0\ ,\label{h4}\\
&&\qquad \qquad \quad \big[{\cW}_{-k-1},{\cG}_{l+1}\big]_q - \big[{\cW}_{-l-1},{\cG}_{k+1}\big]_q - (q-q^{-1})\big({\cW}_{k+1}{\cG}_{l+1}-{\cW}_{l+1}{\cG}_{k+1}\big)=0\ ,\label{h5}\\  
&&\qquad \qquad \quad \big[{\tilde{\cG}}_{l+1},{\cW}_{-k-1}\big]_q - \big[{\tilde{\cG}_{k+1}},{\cW}_{-l-1}\big]_q - (q-q^{-1})\big({\cW}_{k+1}{\tilde{\cG}}_{l+1}-{\cW}_{l+1}{\tilde{\cG}_{k+1}}\big)=0\ ,\label{h6}\\
&&\qquad \qquad \quad \big[{\cW}_{k+2},{\tilde{\cG}}_{l+1}\big]_q - \big[{\cW}_{l+2},{\tilde{\cG}}_{k+1}\big]_q - (q-q^{-1})\big({\cW}_{-k}{\tilde{\cG}}_{l+1}-{\cW}_{-l}{\tilde{\cG}}_{k+1}\big)=0\ .\label{h7}
\eeqa
\end{prop}
\begin{proof}
To show (\ref{h1}), let us consider the first commutator. Expand it using (\ref{qo2}). Combining ${\cW}_0$ and ${\cW}_{l+1}$ using (\ref{qo1}), one finds:
\beqa
\big[{\cW}_{-k-1},{\cW}_{l+1}\big] &=& \frac{q}{\rho(q+q^{-1})}\big(\tilde{{\cG}}_{l+1}{\cG}_{k+1}-\tilde{{\cG}}_{k+1}{\cG}_{l+1}\big)\nonumber\\ 
&&+ \frac{q^{-1}}{\rho(q+q^{-1})}\big({\cG}_{l+1}\tilde{{\cG}}_{k+1}-{\cG}_{k+1}\tilde{{\cG}}_{l+1}\big) +
\big[{\cW}_{-l-1},{\cW}_{k+1}\big]\ .\nonumber
\eeqa
Then, using (\ref{qo5}) and (\ref{qo11}) one obtains (\ref{h1}).
  
Consider now (\ref{h2}). Introduce (\ref{g1}) in the last commutator, and expand using (\ref{qo2}) and (\ref{qo3}). Collecting terms and simplifying, one obtains (\ref{h2}). Equation (\ref{h3}), although technically slightly more complicated, is derived along the same line.

To show (\ref{h4})-(\ref{h7}), the same procedure will be used so we only explain (\ref{h4}). Consider the two commutators and expand using (\ref{qo3}). Then, using (\ref{qo8}) and (\ref{qo11}), one verifies that (\ref{h4}) is indeed satisfied. 
\end{proof}

By analogy with Drinfeld's construction, we are now looking for an infinite dimensional set of elements of an algebra in terms of which the currents $\cW_\pm(u)$, $\cG_\pm(u)$ can be expanded. According to the structure of the equations (\ref{ec1})-(\ref{ec16}) defining the current algebra - in particular the dependence in the formal variable $U,V$ - we obtain the second main result of the paper. 

\begin{thm}\label{th2} Define the formal variable $U=(qu^2+q^{-1}u^{-2})/(q+q^{-1})$. Let $\Psi: O_q(\widehat{sl_2}) \mapsto {\cal A}_q $ be the map defined by
\begin{align}
{\cW}_+(u)=\sum_{k\in {\mathbb Z}_+}{\cW}_{-k}U^{-k-1} \ , \quad {\cW}_-(u)=\sum_{k\in {\mathbb Z}_+}{\cW}_{k+1}U^{-k-1} \ ,\label{c1}\\
 \quad {\cG}_+(u)=\sum_{k\in {\mathbb Z}_+}{\cG}_{k+1}U^{-k-1} \ , \quad {\cG}_-(u)=\sum_{k\in {\mathbb Z}_+}\tilde{{\cG}}_{k+1}U^{-k-1} \ .\label{c2}
\end{align}
Then, $\Psi$ is an algebra isomorphism between  $O_q(\widehat{sl_2})$ and ${\cal A}_q$.
\end{thm}
\begin{proof} Plugging (\ref{c1}), (\ref{c2}) into (\ref{ec1})-(\ref{ec16}), expanding and identifying terms of same order in $U^{-k}V^{-l}$ one finds all defining relations (\ref{qo1})-(\ref{qo11}), together with the set of higher relations (\ref{h1})-(\ref{h7}). From Proposition \ref{prop3}, it follows that the sixteen independent algebraic relations (\ref{qo1})-(\ref{qo11}) are sufficient i.e. the map is surjective. The currents being analytic in the variable $U\in {\mathbb C}$, according to Cauchy's theorem any element of ${\cal A}_q$ is uniquely determined from 
the currents using contour integrals. The injectivity of the map follows, which completes the proof.
\end{proof}

It is important to stress that in \cite{BK}, commutation relations among the so-called transfer matrix were used to derive some of the relations (\ref{qo1})-(\ref{qo11}). However, the derivation described above uses solely the reflection equation algebra. Consequently, this theorem establishes a rigorous proof of the relations conjectured in \cite{BK}. In addition, for the case of the reflection equation algebra with the $U_q(\widehat{sl_2})$ $R$-matrix it shows that the presentation $\{{\cW}_{-k},{\cW}_{k+1},{\cG}_{k+1},{\tilde{\cG}}_{k+1}|k\in {\mathbb Z}_+\}$ is the ``boundary'' analogue of Drinfeld's one.

\section{intertwiner of the $q-$Onsager (tridiagonal) algebra\\
 and the reflection equation}
The purpose of this Section is to exhibit an intertwiner $K(u)$ of the $q-$Onsager algebra, to show its uniqueness and that it coincides exactly with the solution $K(u)$ of the reflection equation (\ref{RE}). The final aim is actually to establish the isomorphism between the new current algebra and the $q-$Onsager algebra. Although the reader may be familiar with the ideas of \cite{J1}, it will be useful to first recall some well-known results. Indeed, the procedure we follow to construct the intertwiner is analogous to the one described in \cite{J1}. In the context of quantum integrable systems, note that for finite dimensional representations intertwiners have already been obtained along the same line in \cite{NM,DMS,Nep,DelG,DM}.\vspace{3mm}

{\bf a. The $R-$matrix as an intertwiner of $U_q(\widehat{sl_2})$} \cite{J1}.\\ 
In \cite{J1}, Jimbo pointed out that intertwiners $R$ of quantum loop algebras lead to trigonometric solutions of the quantum Yang-Baxter equation (\ref{YB}). Any tensor product of two evaluation representations with generic evaluation parameters $u$ and $v$ being indecomposable, by Schur's lemma the solution $R$ is unique up to an overall scalar factor. In particular, considering the quantum affine algebra $U_{q}(\widehat{sl_2})$ the construction of the solution $R(u)$ given by (\ref{R}) goes as follow. 

First, we need to recall the realization of the quantum affine algebra $U_{q}(\widehat{sl_2})$ in the Chevalley presenation $\{H_j,E_j,F_j\}$, $j\in \{0,1\}$ (see e.g \cite{Chari}): 
\begin{defn} Define the extended Cartan matrix $\{a_{ij}\}$ ($a_{ii}=2$,\ $a_{ij}=-2$ for $i\neq j$). The quantum affine algebra $U_{q}(\widehat{sl_2})$ is generated by the elements $\{H_j,E_j,F_j\}$, $j\in \{0,1\}$ which satisfy the defining relations
\beqa [H_i,H_j]=0\ , \quad [H_i,E_j]=a_{ij}E_j\ , \quad
[H_i,F_j]=-a_{ij}F_j\ ,\quad
[E_i,F_j]=\delta_{ij}\frac{q^{H_i}-q^{-H_i}}{q-q^{-1}}\
\nonumber\eeqa
together with the $q-$Serre relations
\beqa [E_i,[E_i,[E_i,E_j]_{q}]_{q^{-1}}]=0\ ,\quad \mbox{and}\quad
[F_i,[F_i,[F_i,F_j]_{q}]_{q^{-1}}]=0\ . \label{defUq}\eeqa
The sum ${\it K}=H_0+H_1$ is the central element of the algebra. The
Hopf algebra structure is ensured by the existence of a
comultiplication $\Delta: U_{q}(\widehat{sl_2})\mapsto U_{q}(\widehat{sl_2})\otimes U_{q}(\widehat{sl_2})$, antipode ${\cal S}: U_{q}(\widehat{sl_2})\mapsto U_{q}(\widehat{sl_2})$ 
and a counit ${\cal E}: U_{q}(\widehat{sl_2})\mapsto {\mathbb C}$ with
\beqa \Delta(E_i)&=&E_i\otimes q^{-H_i/2} +
q^{H_i/2}\otimes E_i\ ,\nonumber \\
 \Delta(F_i)&=&F_i\otimes q^{-H_i/2} + q^{H_i/2}\otimes F_i\ ,\nonumber\\
 \Delta(H_i)&=&H_i\otimes I\!\!I + I\!\!I \otimes H_i\ ,\label{coprod}
\eeqa
\beqa {\cal S}(E_i)=-E_iq^{-H_i}\ ,\quad {\cal S}(F_i)=-q^{H_i}F_i\ ,\quad {\cal S}(H_i)=-H_i \qquad {\cal S}({I\!\!I})=1\
\label{antipode}\nonumber\eeqa
and\vspace{-0.3cm}
\beqa {\cal E}(E_i)={\cal E}(F_i)={\cal
E}(H_i)=0\ ,\qquad {\cal E}({I\!\!I})=1\
.\label{counit}\nonumber\eeqa
\end{defn}
Note that the opposite coproduct $\Delta'$ can be similarly defined with $\Delta'\equiv \sigma
\circ\Delta$ where the permutation map $\sigma(x\otimes y
)=y\otimes x$ for all $x,y\in U_{q}(\widehat{sl_2})$ is used.\vspace{2mm} 

Then, by definition the intertwiner $R(u/v):{\cal V}_u\otimes {\cal V}_v\mapsto {\cal V}_v\otimes {\cal V}_u$ between two fundamental $U_{q}(\widehat{sl_2})-$evaluation representations obeys
\beqa
R(u/v)(\pi_u\times\pi_v)\big[\Delta(x)\big]= (\pi_u\times\pi_v)\big[\Delta'(x)\big]R(u/v) \qquad \forall x\in U_{q}(\widehat{sl_2})\ ,\label{condR}
\eeqa
where the (evaluation) endomorphism  $\pi_u: U_{q}(\widehat{sl_2}) \mapsto \mathrm{End}({\cal V}_u)$ is chosen such that $({\cal V}\equiv{\mathbb C}^2)$
\beqa 
&&\pi_u[E_1]= uq^{1/2}\sigma_+\ , \qquad \ \ \ \ \ \pi_u[E_0]= uq^{1/2}\sigma_-\ , \nonumber\\
&&\pi_u[F_1]=
u^{-1}q^{-1/2}\sigma_-\ ,\qquad \pi_u[F_0]= u^{-1}q^{-1/2}\sigma_+\ ,\nonumber\\
\ &&\pi_u[q^{H_1}]= q^{\sigma_3}\ ,\qquad \qquad \quad  \ \pi_u[q^{H_0}]=
q^{-\sigma_3}\ \label{evalrep}
\eeqa
in terms of the Pauli matrices $\sigma_\pm,\sigma_3$:
\beqa
\sigma_+=\left(
\begin{array}{cc}
 0    & 1 \\
 0 & 0 
\end{array} \right) \ ,\qquad
\sigma_- =\left(
\begin{array}{cc}
 0    & 0 \\
 1 & 0 
\end{array} \right) \ ,\qquad
\sigma_3 =\left(
\begin{array}{cc}
 1    & 0 \\
 0 & -1 
\end{array} \right) \ .\label{Pauli}\nonumber
\eeqa

As one can easily check, the matrix $R(u)$ given by (\ref{R}) indeed satisfies the required conditions (\ref{condR}). The tensor product ${\cal V}_u\otimes {\cal V}_v$ being indecomposable with respect to $U_{q}(\widehat{sl_2})$ evaluation representations for generic evaluation parameters $u,v$, the intertwiner $R(u)$ is unique (up to an overall scalar factor). As a consequence, it automatically satisfies the Yang-Baxter equation (\ref{YB}) which may be depicted by the following commutative diagram setting $w=1$:
\begin{equation}
\begin{CD}
  {\cal V}_u\otimes {\cal V}_v\otimes {\cal V}_w
  @>{R(u/v)\,\otimes\,\id}>>
  {\cal V}_v\otimes {\cal V}_u\otimes {\cal V}_w
  @>{id\,\otimes\, R(u/w)}>>
  {\cal V}_v\otimes {\cal V}_w\otimes {\cal V}_u
  \\
  @VV{\id\otimes R(v/w)}V
  @.
  @V{R(v/w)\otimes\id}VV
  \\
  {\cal V}_u\otimes {\cal V}_w\otimes {\cal V}_v
  @>{R(u/w)\,\otimes\,\id}>>
  {\cal V}_w\otimes {\cal V}_u\otimes {\cal V}_v
  @>{id\,\otimes\, R(u/v)}>>
  {\cal V}_w\otimes {\cal V}_v\otimes {\cal V}_u
\end{CD}
\end{equation}

\vspace{7mm}

{\bf b. The $K-$matrix as an intertwiner of ${\mathbb T}$}.\\
Tridiagonal algebras have been introduced and studied in
\cite{Ter93,Ter01,Ter03}, where they first appeared in the context
of $P-$ and $Q-$polynomial association schemes. A tridiagonal algebra is an associative algebra with unit which consists of two generators ${\textsf A}$ and ${\textsf A}^*$ called the standard generators. In general, the defining relations depend on five scalars $\rho,\rho^*,\gamma,\gamma^*$ and $\beta$. In the following, we will focus on the {\it reduced} parameter sequence $\gamma=0,\gamma^*=0$, $\beta=q^2+q^{-2}$ and $\rho=\rho^*$  which exhibits all interesting properties that can be extended to more general parameter sequences.
We call the corresponding algebra the $q-$Onsager algebra denoted ${\mathbb T}$, in view of its closed relationship with the Onsager algebra \cite{Ons44} and the Dolan-Grady relations \cite{DG}. In particular, the isomorphism between the Onsager and Dolan-Grady  algebraic structures has been studied in \cite{Pe,APMPT,Da} and shown explicitly in \cite{DateRoan}.
\begin{defn}[see also \cite{Ter03}]
The $q-$Onsager algebra $\mathbb{T}$ is the associative algebra with unit and standard generators $\textsf{A},\textsf{A}^*$ subject to the following relations
\beqa
[\textsf{A},[\textsf{A},[\textsf{A},\textsf{A}^*]_q]_{q^{-1}}]=\rho[\textsf{A},\textsf{A}^*]\
,\qquad
[\textsf{A}^*,[\textsf{A}^*,[\textsf{A}^*,\textsf{A}]_q]_{q^{-1}}]=\rho[\textsf{A}^*,\textsf{A}]\
\label{qDG} . \eeqa
\end{defn}
\begin{rem} For $\rho=0$ the relations (\ref{qDG})
reduce to the $q-$Serre relations of $U_{q}(\widehat{sl_2})$. For $q=1$, $\rho=16$ they coincide with the Dolan-Grady relations \cite{DG}.
\end{rem}

By analogy with the construction described above for the $R-$matrix and along the lines described in \cite{DM,DelG}, an intertwiner for ${\mathbb T}$ can be easily constructed. Before, we need to introduce the concept of comodule algebra using the analogue of the Hopf's algebra coproduct action called the coaction map. 
\begin{defn}[\cite{Chari}] Given a Hopf algebra
${\cal H}$ with comultiplication $\Delta$ and counit ${\cal E}$, ${\cal I}$ is
called a left ${\cal H}-$comodule if there exists a left
coaction map $\delta:\ \ {\cal I}\rightarrow {\cal H}
\otimes {\cal I}$ such that
\beqa (\Delta \times id)\circ
\delta=(id\times\delta)\circ\delta\ ,\qquad
({\cal E} \times id)\circ \delta \cong id \
.\label{defcoaction}\eeqa
Right ${\cal H}-$comodules are defined similarly.
\end{defn}
\begin{prop}[see also \cite{B}] 
Let $k_\pm\in{\mathbb C}^*$ and set $\rho\equiv k_+k_-(q+q^{-1})^2$. The q-Onsager algebra ${\mathbb T}$ is a left $U_{q}(\widehat{sl_2})-$comodule algebra with coaction map $\delta: {\mathbb T}\rightarrow U_{q}(\widehat{sl_2})\otimes {\mathbb T}$
%counit $\epsilon: {\mathbb T}\mapsto {\mathbb C}$ and antipode 
%$s: {\mathbb T}\mapsto {\mathbb T}$ 
such that
\beqa
\delta({\textsf A})&=& (k_+E_1q^{H_1/2} + k_-F_1q^{H_1/2})\otimes 1 + q^{H_1} \otimes {\textsf A}\ ,\nonumber\\
\delta({\textsf A}^*)&=& (k_-E_0q^{H_0/2} + k_+F_0q^{H_0/2})\otimes 1 + q^{H_0} \otimes {\textsf A}^*\ .\label{deltadef}
%\epsilon({\textsf A})&=&\epsilon({\textsf A}^*)=0\ ,\nonumber\\
%s({\textsf A})&=&{\textsf A}^*\ ,\qquad s({\textsf A}^*)={\textsf A}\ .
\eeqa
\end{prop}
\begin{proof}
The verification of the comodule algebra axioms (\ref{defcoaction}) is immediate using (\ref{coprod}). One also has to check that $\delta$ is an algebra homomorphism i.e $\delta({\textsf A}),\delta({\textsf A}^*)$ satisfy (\ref{qDG}). This calculation is rather long but straightforward, so we omit the details (see also \cite{B2,TD}).
\end{proof}
Having identified such a coaction map, we are now in position to consider an intertwiner relating representations of ${\mathbb T}$, a key ingredient in relating the $q-$Onsager algebra and the reflection equation algebra.
\begin{prop}\label{propT}
Let $\pi_u: U_{q}(\widehat{sl_2}) \mapsto \mathrm{End}({\cal V}_u)$ be the evaluation endomorphism for ${\cal V}\equiv {\mathbb C}^2$. Let $W$ denote a vector space over ${\mathbb C}$ on which the elements of ${\mathbb T}$ act. There exists an intertwinner 
\beqa
K(u): {\cal V}_u\otimes W \mapsto {\cal V}_{u^{-1}}\otimes W \nonumber
\eeqa
such that
\beqa
K(u)(\pi_u\times id)\big[\delta(a)\big]= (\pi_{u^{-1}}\times id)\big[\delta(a)\big]K(u)\ , \qquad \forall a\in {\mathbb T}\ .\label{condK}
\eeqa
It is unique (up to an overall scalar factor), and it satisfies the reflection equation (\ref{RE}).
\end{prop}
\begin{proof}
First, let us identify one solution of (\ref{condK}). By definition, ${\cal V}_u$ is a two-dimensional vector space. Then $K(u)$ is a $2\times 2$ matrix, which entries are formal power series in the variable $u$ in view of (\ref{evalrep}) and (\ref{deltadef}). Define
\beqa
K(u)=\left(
\begin{array}{cc}
 A(u)    & B(u) \\
 C(u) & D(u) 
\end{array} \right)\ .\nonumber
\eeqa
Replacing $K(u)$ in (\ref{condK}), we find that the entries must satisfy the following system of equations
\beqa
&&\big[{\textsf A},A(u)\big] = q^{-1}u^{-1}\big(k_-B(u) - k_+C(u)\big)\ ,\nonumber\\
&&\big[{\textsf A},D(u)\big] = -qu\big(k_-B(u) - k_+ C(u)\big)\ ,\label{al}\\
&&\big[{\textsf A},B(u)\big]_q = k_+\big(uA(u) - u^{-1}D(u)\big)\ ,\nonumber\\
&&\big[{\textsf A},C(u)\big]_{q^{-1}} = -k_-\big(uA(u)-u^{-1}D(u) \big)\ \nonumber
\eeqa
and similar relations for ${\textsf A}^*$, provided one substitutes $q\rightarrow q^{-1},u\rightarrow u^{-1}$ in (\ref{al}). Then, using (\ref{g1}) in (\ref{qo2}) for $k=0$ it is easy to notice that the defining relations (\ref{qDG}) are nothing but (\ref{qo4}) for $k=0,l=1$, provided we consider the following homomorphism
\beqa
{\textsf A}\mapsto {\cW}_0\ ,
\qquad {\textsf A}^*\mapsto {\cW}_1\ .\label{h}
\eeqa
Now, identify the entries of $K(u)$ with (\ref{m1})-(\ref{m4}). Expanding and using the defining relations (\ref{qo1})-(\ref{qo3}) of the algebra ${\cal A}_q$, it is easy to check (\ref{al}) as well as all other relations for ${\textsf A}^*$. So, at least one solution $K(u)$ exists and it is written in terms of elements of ${\cal A}_q$. For generic $u$, the tensor product $\mathrm{End}({\cal V}_u)\otimes W$ is not decomposable with respect to ${\mathbb T}$ representations. By Schur's lemma, this means that given $W$, the solution to the intertwining relation (\ref{condK}) is unique (up to an overall scalar factor). It remains to show that $K(u)$ satisfying (\ref{condK}) is automatically a solution of the reflection equation algebra (\ref{RE}). To this end, let us recall that $K(u): {\cal V}_u\otimes W \mapsto {\cal V}_{u^{-1}}\otimes W$ and $R(u/v):{\cal V}_u\otimes {\cal V}_v\mapsto {\cal V}_v\otimes {\cal V}_u$. Then, the proof that this solution $K(u)$ satisfies the reflection equation (\ref{RE}) follows from the commutativity of the following diagram (up to an overall scalar factor):
\begin{equation}
\begin{CD}
  {\cal V}_u\otimes {\cal V}_v\otimes W
  @>{id\,\otimes\, K(v)}>>
  {\cal V}_u\otimes {\cal V}_{v^{-1}}\otimes W
  @>{R(uv)\,\otimes\, id}>>
  {\cal V}_{v^{-1}}\otimes {\cal V}_u \otimes W
  \\
  @VV{R(u/v)\otimes id}V
  @.
  @V{id\otimes K(u)}VV
  \\
  {\cal V}_v\otimes {\cal V}_u\otimes W
  @.
  \qquad
  @. 
  {\cal V}_{v^{-1}}\otimes {\cal V}_{u^{-1}} \otimes W
  \\
  @VV{id\otimes K(u)}V
  @.
  @V{R(u/v)\otimes id}VV
  \\
  {\cal V}_v\otimes {\cal V}_{u^{-1}}\otimes W
  @>{R(uv)\,\otimes\,\id}>>
  {\cal V}_{u^{-1}}\otimes {\cal V}_v\otimes W
  @>{id\,\otimes\, K(v)}>>
  {\cal V}_{u^{-1}}\otimes {\cal V}_{v^{-1}}\otimes W  \nonumber
\end{CD}
\end{equation}
\end{proof}

\vspace{1mm}

Combining previous results, we obtain the third main result of the paper:
\begin{thm}\label{th3} The $q-$Onsager algebra ${\mathbb T}$ and the current algebra $O_q(\widehat{sl_2})$ are isomorphic.
\end{thm}
\begin{proof}
According to Proposition \ref{propT}, $K(u)$ with (\ref{m1})-(\ref{m4}) is the unique intertwiner of ${\mathbb T}$ satisfying (\ref{condK}). Also, it satisfies the reflection equation algebra (\ref{RE}). So, $K(u)$ establishes the isomorphism between ${\mathbb T}$ and the reflection equation algebra (\ref{RE}) for the $U_q(\widehat{sl_2})$ $R-$matrix.
Theorem \ref{th1} then establishes the isomorphism between the reflection equation algebra (\ref{RE}) and $O_q(\widehat{sl_2})$, which supports the claim. 
\end{proof}
Although the isomorphism between ${\mathbb T}$ and $O_q({\widehat{sl_2}})\cong {\cal A}_q$ is now established, an interesting problem remains to construct an explicit homomorphism from ${\cal A}_q$ to ${\mathbb T}$, i.e. to write all higher elements of ${\cal A}_q$ solely in terms of ${\cW}_0,{\cW}_1$. This problem will be considered elsewhere.\vspace{1mm} 

To conclude, the $q-$Onsager algebra ${\mathbb T}$ admits two different realizations: one [see Proposition \ref{propT}] in terms of the reflection equation algebra for the $U_q({\widehat{sl_2}})$ $R-$matrix and another one in terms [see Theorems \ref{th1}, \ref{th2}, \ref{th3}] of the current algebra $O_q({\widehat{sl_2}})\cong {\cal A}_q$.   
Previous results are resumed by the picture below.
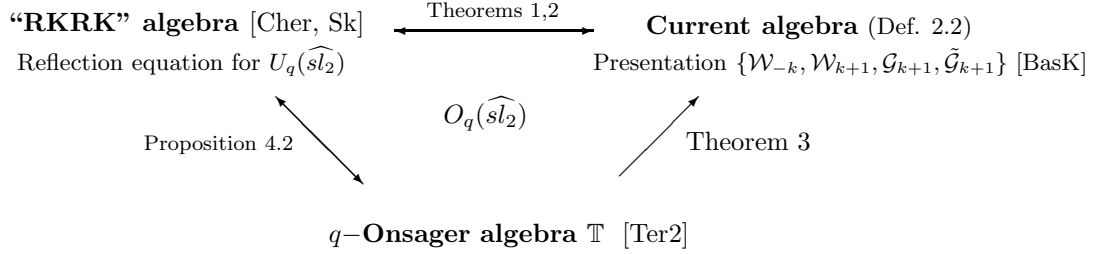
\begin{figure}[ht!]
\begin{center}
\begin{picture}(390,90)
   \put(5,60){\shortstack[1]{\ \bf``RKRK'' algebra \cite{Cher84,Skly88} \\
        \small Reflection equation for $U_q(\widehat{sl_2})$}}
   %\put(150,75){\vector(1,0){70}}
   \put(225,75){\vector(-2,0){70}}
   \put(160,75){\vector(2,0){70}}
   \put(150,80){\shortstack[l]{ \\
                                 \footnotesize \qquad Theorems \ref{th1},\ref{th2}}}
   \put(230,60){\shortstack[l]{\quad \quad  {\bf Current algebra} \small (Def. \ref{defnCA})  \\
                                     \small Presentation \small $\{{\cW}_{-k},{\cW}_{k+1},{\cG}_{k+1},{\tilde{\cG}}_{k+1}\}$ \cite{BK}}} 
   \put(140,20){\vector(-1,1){30}}
   \put(113,47){\vector(1,-1){30}}
   \put(170,40){\shortstack[l]{\ {\bf $O_q({\widehat{sl_2}})$}}}                             
   \put(60,30){\shortstack[l]{\footnotesize Proposition \ref{propT}}}
   
   %\put(270,50){\vector(-1,-1){30}}
   %\put(225,30){\shortstack[l]{\footnotesize \qquad \qquad \cite{Be}}}                             
   \put(240,20){\vector(1,1){30}}
   \put(265,30){\shortstack[l]{Theorem \ref{th3}}}

   \put(265,23){\shortstack[l]{}}
   \put(130,-5){\shortstack[l]{{\bf $q-$Onsager algebra ${\mathbb T}$} \ \cite{Ter03}}}
\end{picture}
\end{center}
\vspace{5mm}
\caption{\label{fig1}
An algebraic scheme for $O_q({\widehat{sl_2}})$}
\end{figure} 

\vspace{0.7cm}

%\newpage
%
\noindent{\bf Acknowledgements:} Part of this work has been supported by the ANR Research project ``{\it Boundary integrable models: algebraic structures and correlation functions}'', contract number JC05-52749. P.B thanks S. Pakuliak for detailed explanations and notes about Drinfeld's construction at the early stage of this work, as well as P. Terwilliger for reading the manuscript and helpful comments.

\vspace{0.2cm}

\end{document}